\documentclass[conference]{IEEEtran}
\usepackage{amsmath,amssymb}
\usepackage{subfigure}
\usepackage{graphicx,graphics,color,psfrag}
\usepackage{cite,balance}
\usepackage{algorithm}
\usepackage{accents}
\usepackage{amsthm}
\usepackage{bm}
\usepackage{url}
\usepackage{algorithmic}
\usepackage[english]{babel}
\usepackage{multirow}
\usepackage{enumerate}
\usepackage{cases}
\usepackage{stfloats}
\usepackage{dsfont}
\usepackage{color,soul}
\usepackage{amsfonts}
\usepackage{tcolorbox}
\usepackage{amsmath}
\usepackage{float}

\usepackage{cite,graphicx,amsmath,amssymb}
\usepackage{fancyhdr}
\usepackage{hhline}
\usepackage{graphicx,graphics}
\usepackage{array,color}
\usepackage{amsmath}
\usepackage{amsthm}
\usepackage[flushleft]{threeparttable}

\theoremstyle{definition}

\newtheorem{proposition}{Proposition}

\ifCLASSINFOpdf

\else

\fi

\begin{document}

\title{Joint Transmit and Reflective Beamforming for IRS-Assisted Integrated Sensing and Communication}

\author{
\IEEEauthorblockN{Xianxin~Song\IEEEauthorrefmark{1},  Ding~Zhao\IEEEauthorrefmark{2}, Haocheng~Hua\IEEEauthorrefmark{1}, Tony~Xiao~Han\IEEEauthorrefmark{3}, Xun~Yang\IEEEauthorrefmark{3}, and~Jie~Xu\IEEEauthorrefmark{1}}

\IEEEauthorblockA{\IEEEauthorrefmark{1}SSE and FNii, 
The Chinese University of Hong Kong, Shenzhen, China}
\IEEEauthorblockA{\IEEEauthorrefmark{2}The College of Information Science and Electronic Engineering, Zhejiang University,  China}
\IEEEauthorblockA{\IEEEauthorrefmark{3}Wireless Technology Lab, 2012 Laboratories, Huawei, China}
Email: xianxinsong@link.cuhk.edu.cn, 3180103736@zju.edu.cn, haochenghua@link.cuhk.edu.cn,\\
 \{tony.hanxiao, david.yangxun\}@huawei.com, xujie@cuhk.edu.cn
}

\maketitle

\begin{abstract}
This paper studies an intelligent reflecting surface (IRS)-assisted integrated sensing and communication (ISAC) system, in which one IRS with a uniform linear array (ULA) is deployed to not only assist the wireless communication from a multi-antenna base station (BS) to a single-antenna communication user (CU), but also create virtual line-of-sight (LoS) links for sensing potential targets at areas with LoS links blocked. We consider that the BS transmits combined information and sensing signals for ISAC. Under this setup, we jointly optimize the transmit information and sensing beamforming at the BS and the reflective beamforming at the IRS, to maximize the IRS’s minimum beampattern gain towards the desired sensing angles, subject to the minimum signal-to-noise ratio (SNR) requirement at the CU and the maximum transmit power constraint at the BS. Although the formulated SNR-constrained beampattern gain maximization problem is non-convex and difficult to solve, we present an efficient algorithm to obtain a high-quality solution by using the techniques of alternating optimization and semi-definite relaxation (SDR). Numerical results show that the proposed joint beamforming design achieves improved sensing performance while ensuring the communication requirement as compared to benchmarks without such joint optimization. It is also shown that the use of dedicated sensing beams is beneficial in enhancing the performance for IRS-assisted ISAC.
\end{abstract}


\IEEEpeerreviewmaketitle

\section{Introduction}
Integrated sensing and communication (ISAC) has been recognized as one of the candidate key technologies for beyond-fifth-generation (B5G) and sixth-generation (6G)  wireless networks to enable environment-aware applications such as auto-driving, industrial automation, and mixed reality, in which the wireless infrastructures and spectrum resources are reused for radar sensing, localization, and imaging (see, e.g., \cite{8999605,9468975,liu2021integrated,9606831} and the references therein). In order to efficiently provide both  sensing and communication services, various designs on ISAC system architectures, ISAC waveform, and transmit beamforming have  been proposed in prior works (see, e.g., \cite{9246715,9415119,hua,ren2021optimal}). 

Despite the recent research progress, ISAC networks face new technical challenges, because the communication and sensing systems deal with the multipath wireless channels in different ways. In wireless communications, both line-of-sight (LoS) and non-LoS (NLoS) links in multipath channels can be exploited beneficially  to enhance the communication rate by providing more degrees of freedom. By contrast, in radar sensing, only LoS links are utilized for information extraction, by treating the NLoS links as harmful interference or clutters. As a result, how to provide ubiquitous sensing coverage for areas with LoS links blocked remains a challenge, especially in the scenario with dense obstacles such as buildings and trees. 

To resolve this issue, the intelligent reflecting surface (IRS) or reconfigurable intelligent surface (RIS) has emerged as a viable solution \cite{9140329,8811733,aubry2021reconfigurable,9361184,Stefano,9264225,9416177,9364358}. By properly adjusting the phase shifts of digitally-controlled reflecting elements, the IRS can help reconfigure the wireless propagation environment to create virtual LoS links for sensing targets without LoS connections. While there is rich literature on the IRS-assisted wireless communications (see, e.g., \cite{9140329,8811733} and the references therein), only several recent works investigated the IRS-assisted wireless sensing  \cite{aubry2021reconfigurable,Stefano,9361184}, radar-communication coexistence \cite{9264225}, and ISAC \cite{9416177,9364358}. For instance, the authors in \cite{9416177} studied an IRS-assisted ISAC system with one base station (BS) and multiple communication users (CUs), in which only the communication was assisted by the IRS while the sensing was based on the direct LoS links. Furthermore, the authors in \cite{9364358} considered a simplified ISAC setup with one BS, one CU, and one target, in which the signal-to-noise ratio (SNR) of radar is maximized while ensuring the SNR at the CU. To the best of our knowledge, how to utilize IRS to enable the NLoS {\it multi-target} sensing and assist the communication at the same time has not been investigated in the literature yet.

This paper considers an IRS-assisted ISAC system consisting of one multi-antenna BS, one IRS with a uniform linear array (ULA), one single-antenna CU, and multiple potential targets at the NLoS areas of the BS. In particular, we consider that the BS sends one information beam combined with multiple dedicated sensing beams to facilitate ISAC. Because the direct LoS links from the BS to the potential targets are not available, the BS can only sense these targets via the virtual LoS link reflected by the IRS. By contrast, the BS can exploit both its direct and reflected LoS/NLoS links with the CU for efficient communication. 

In particular, our objective is to maximize the IRS's minimum beampattern gain towards the desired sensing angles, while ensuring the minimum SNR requirement at the CU and the maximum transmit power constraint at the BS. Although the formulated SNR-constrained beampattern gain maximization problem is non-convex and difficult to solve in general, we present an efficient algorithm to obtain a high-quality solution by using the techniques of alternating optimization and semi-definite relaxation (SDR). Numerical results show that the proposed joint beamforming design achieves improved sensing performance while ensuring the communication requirement, as compared to heuristically designed benchmarks. It is also shown that our proposed design with dedicated sensing beams significantly outperforms the benchmark scheme without using sensing beams (by only reusing information beams for sensing), which reveals the importance of dedicated sensing beams in IRS-assisted ISAC systems. 

{\it Notations}: Boldface letters refer to vectors (lower case) or matrices (upper case). For a square matrix $\bm S$, $\mathrm {tr}(\bm S)$ and $\bm S^{-1}$ denote its trace and inverse, respectively, while $\bm S \succeq \mathbf{0}$ means that $\bm S$ is positive semidefinite. For an arbitrary size matrix $\bm M$, $\mathrm {rank}(\bm M)$, $\bm M^{\mathrm {H}}$, and $\bm M^{\mathrm {T}}$ denote its rank, conjugate transpose, and transpose, respectively. The distribution of a circularly symmetric complex Gaussian (CSCG) random vector with mean vector $\bm x$ and covariance matrix $\bm \Sigma$ is denoted by $\mathcal{C N}(\bm{x}, \mathbf{\Sigma})$ and $\sim$ stands for “distributed as”. $\mathbb{C}^{x \times y}$ denotes the space of $x \times y$ complex matrices.  $\mathbb{E}(\cdot)$ denotes the statistical expectation. $\|\bm x\|$ denotes the Euclidean norm of vector $\bm x$. $\mathrm {diag}(a_1,...,a_N)$ denotes a diagonal matrix with diagonal elements $a_1,...,a_N$. $\mathrm {arg}(\bm x)$ denotes a vector with each element being the phase of the corresponding element in $\bm x$. $[\bm{x}]_{(1:N)}$ denotes the vector that contains the first $N$ elements in $\bm x$. 

\section{System Model and Problem Formulation}

\begin{figure}[t]
    \centering
    \includegraphics[width=0.385\textwidth]{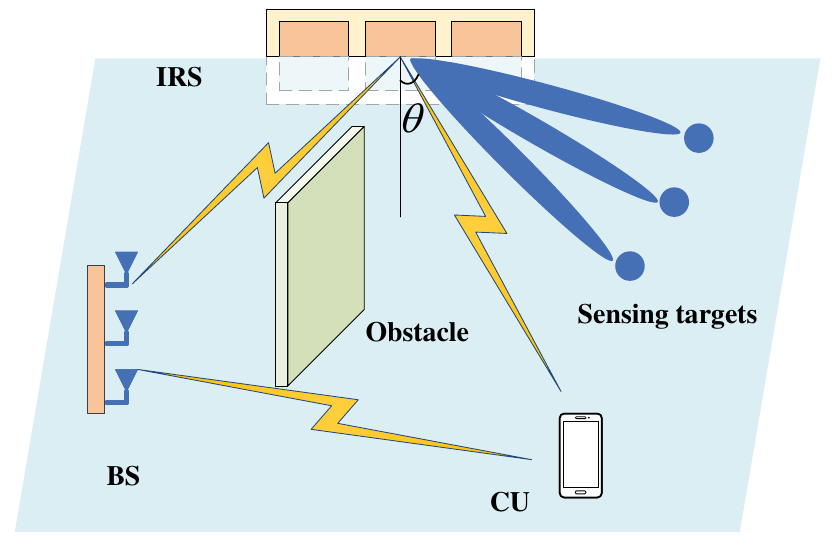}
    \caption{Illustration of the IRS-assisted ISAC system, in which the targets are sensed via the virtual LoS links from the BS to the IRS to the targets.}
    \label{system_model}
\end{figure}

We consider an IRS-assisted ISAC system as shown in Fig.~\ref{system_model}, which consists of one BS with $M>1$ transmit antennas, one CU with one receive antenna, an ULA-IRS with $N>1$ elements, and multiple potential targets at the NLoS areas of the BS. It is assumed that the BS perfectly knows the global channel state information (CSI) and the targets' rough directions for the purpose of initial investigation.

To facilitate ISAC, we consider that the BS uses the transmit beamforming to send both  information and dedicated radar sensing signals. Let $s$ denote the information signal desired by the CU, and $\bm{w} \in \mathbb{C}^{M \times 1}$ denote the corresponding transmit beamforming vector, where $s$ is assumed to be a random variable with zero mean and unit variance. Let $\bm{x}_0 \in \mathbb{C}^{M\times 1}$ denote the dedicated sensing signal, which is a randomly generated sequence independent of $s$, with zero mean and covariance matrix $\bm R_0 \triangleq \mathbb{E}(\bm{x}_0\bm{x}_0^{\mathrm {H}}) \succeq \bm{0}$. Notice that we consider the general multi-beam sensing signal transmission with $0 \le \mathrm {rank}(\bm R_0) \le M$, in order to provide more degrees of freedom for sensing. Here, $\mathrm {rank}(\bm R_0)$ corresponds to the number of sensing beams sent by the BS, each of which can be obtained via the eigenvalue decomposition (EVD) of $\bm R_0$. By combining the information and sensing signals, the transmitted signal $\bm{x}$ by the BS is
\begin{equation}
  \bm{x} = \bm{w} s + \bm{x}_0.
\end{equation}
Let $P_0$ denote the maximum transmit power at the BS. We have the transmit power constraint as 
\begin{equation}\label{equ:sum_power_constr}
  \mathbb{E}(\| \bm x\|^2) =\|\bm{w}\|^2 + \mathrm {tr}(\bm R_0) \le P_0.
\end{equation}

First, we consider the information transmission from the BS to the CU.
Let $\bm{G} \in \mathbb{C}^{N \times M}$, $\bm{h}_{\text{r}} \in \mathbb{C}^{N \times 1}$, and $\bm{h}_{\text{d}} \in \mathbb{C}^{M \times 1}$ denote the channel matrix of the BS-IRS link, and the channel vectors of the IRS-CU and BS-CU links, respectively. Notice that the channel matrix/vectors $\bm G$, $\bm h_{\mathrm {r}}$, and $\bm h_{\mathrm {d}}$ may contain both LoS and NLoS signal paths in the corresponding links. Furthermore, let $\phi_n \in (0, 2\pi]$ denote the phase shift of reflecting element $n \in \{1,2,\cdots,N\}$ at the IRS, and $\bm{\Phi} = \mathrm {diag}(e^{j\phi_1},...,e^{j\phi_N})$  denote the corresponding reflection beamforming matrix. 
By combining the signals transmitted through the direct BS-CU link and the reflected BS-IRS-CU link, the received signal at the CU is
\begin{equation}
    y = (\bm{h}_{\text{r}}^{\mathrm {H}} \bm{\Phi} \bm{G} + \bm{h}_{\text{d}}^{\mathrm {H}}) (\bm{w}s+\bm x_0)+n,
\end{equation}
where $n \sim \mathcal{CN}(0,\sigma^2)$ denotes the additive white Gaussian noise (AWGN) at the CU receiver. 
Notice that at the CU receiver, the reception of information signal $s$ may suffer from the interference caused by sensing signal $\bm x_0$. Nevertheless, as $\bm x_0$ can be generated offline and thus known by the CU prior to the transmission, the CU can pre-cancel the interference from $\bm x_0$ to facilitate the decoding of $s$\cite{hua}. After such processing, the received SNR at the CU is
\begin{equation}\label{eq:SNR}
    \gamma=\frac{|(\bm{h}_{\text{r}}^{\mathrm {H}} \bm{\Phi} \bm{G} + \bm{h}_{\text{d}}^{\mathrm {H}})\bm{w}|^2}{\sigma^2}.
\end{equation}

Next, we consider the radar sensing towards the potential targets at the NLoS areas of the BS. In this case, we use the virtual LoS links created by the IRS's reflection to sense them.
Accordingly, we adopt the IRS's beampattern gain towards the desired sensing angles as the sensing performance metric. Let $d_{\text{IRS}}$ denote the spacing between consecutive reflecting elements at the IRS and $\lambda$ denote the wavelength. The steering vector at the IRS with angle of departure (AoD)  $\theta$ is
\begin{equation} \label{equ:steering}
  \bm{a}(\theta) = [1,e^{j  \frac{2\pi d_{\text{IRS}}}{\lambda} \sin \theta},...,e^{j  \frac{2\pi (N-1) d_{\text{IRS}}}{\lambda}  \sin\theta}]^{\mathrm {T}}.
\end{equation} 
Similarly in \cite{hua,ren2021optimal}, we consider that both the information signal $s$ and the dedicated sensing signal $\bm x_0$ can be jointly used to illuminate the sensing targets. As a result, the beampattern gain from the IRS towards angle $\theta$ is given as
\begin{equation}\label{eq:beampattern_gain}
\begin{split}
  \mathcal{P}(\theta)=&\mathbb{E}(|\bm{a}^{\mathrm {H}}(\theta)\bm{\Phi} \bm{G} (\bm{w} s + \bm{x}_0)|^2)\\
  =&\bm{a}^{\mathrm {H}}(\theta)\bm{\Phi} \bm{G} (\bm{w} \bm{w}^{\mathrm {H}} + \bm R_0 ) \bm{G}^{\mathrm {H}} \bm{\Phi}^{\mathrm {H}} \bm{a}(\theta).
\end{split}
\end{equation}

We are particularly interested in sensing potential targets at $L$ desired angles with respect to the IRS, denoted by $\theta_1, \cdots, \theta_L$. Let $\mathcal L \triangleq \{1,\cdots, L\}$ denote the set of desired sensing angles. Our objective is to maximize the minimum beampattern gain at the $L$ angles, while ensuring the minimum SNR requirement at the CU and the maximum transmit power constraint at the BS. The SNR-constrained minimum beampattern gain maximization problem is formulated as
\begin{subequations}
  \begin{align}\notag
    \text{(P1)}: \max_{\bm{w},\bm R_0,\bm{\Phi}}&\ \ \min_{l \in \mathcal L} \ \bm{a}^{\mathrm {H}}(\theta_l)\bm{\Phi} \bm{G} (\bm{w} \bm{w}^{\mathrm {H}}+ \bm R_0 ) \bm{G}^{\mathrm {H}} \bm{\Phi}^{\mathrm {H}}\bm{a}(\theta_l)\\\label{eq:SNR_cons}
    \text { s.t. }& \quad \frac{|(\bm{h}_{\text{r}}^{\mathrm {H}} \bm{\Phi} \bm{G} + \bm{h}_{\text{d}}^{\mathrm {H}})\bm{w}|^2}{\sigma^2} \geq \Gamma,\\\label{eq:energy_cons}
    &\quad   \|\bm{w}\|^2  + \mathrm {tr}(\bm R_0) \le P_0,\\\label{eq:R_0}
    &\quad   \bm R_0 \succeq 0,\\\label{eq:phi}
    &\quad   \bm{\Phi} = \mathrm {diag}(e^{j\phi_1},...,e^{j\phi_N}),
  \end{align}
\end{subequations}
where $\Gamma$ denotes the minimum SNR threshold at the CU.
Due to the coupling between the transmit beamformers ($\bm w$ and $\bm R_0$) and the reflective beamformer ($\bm \Phi$), problem $\text{(P1)}$ is highly non-convex and thus very difficult to be optimally solved.

\textit {Remark 1:} It is worth noting that for conventional ISAC systems without IRS, beampatten matching is another widely adopted design criteria for sensing (e.g., \cite{hua,ren2021optimal,8550811}), in which the objective is to minimize the mean squared matching error between the achieved transmit beampattern and a pre-determined beampattern, by allowing the BS to transmit at full power. This design, however, may not be applicable for the IRS-assisted ISAC system of our interest. This is because under the beampattern matching design, the BS may choose to steer the energy orthogonal to the IRS and use the direct LoS link for serving the CU, which may lead to minimized (or even zero) matching error but very small (or even zero) beampattern gains at the IRS that are not desired for sensing.

\section{Proposed Joint Beamforming Solution to Problem $\text{(P1)}$}
This section proposes an efficient algorithm to solve problem $\text{(P1)}$ based on the principle of alternating optimization, in which the transmit beamformers ($\bm w$ and $\bm R_0$) at the BS and the reflective beamformer ($\bm \Phi$) at the IRS are optimized in an alternating manner, by treating the other to be given. The details of the proposed joint beamforming algorithm are summarized in Algorithm $1$. In the following, we focus on the transmit and reflective beamforming optimization, respectively.

\begin{algorithm}[t]
\caption{The proposed joint transmit and reflective beamforming  algorithm.}
\begin{algorithmic}[1]
\STATE Initialize the reflective beamforming matrix $\bm{\Phi}^{(1)}$ and set the iteration number $k= 1$.
\REPEAT
\STATE Solve problem (SDR2.1) under given reflective beamformer $\bm{\Phi}^{(k)}$, and reconstruct an equivalent optimal solution $\bm w^{(k)}$ and $\bm R_0^{(k)}$ using Proposition $1$.

\STATE Solve problem (SDR3.1) under given transmit beamformer $\bm w^{(k)}$ and $\bm R_0^{(k)}$, and reconstruct an approximate rank-one solution $\bm{\Phi}^{(k+1)}$ using Gaussian randomization.

\STATE Update $k=k+1$.
\UNTIL{ The fractional increase of the objective value  is below a threshold $\epsilon>0$ or the maximum number of iterations is
reached.}
\end{algorithmic}
\end{algorithm}

\subsection{Transmit Beamforming Optimization at BS}

First, we optimize the transmit beamformers $\bm w$ and $\bm R_0$ in problem $\text{(P1)}$ under any given reflective beamformer $\bm \Phi$. This problem is formulated as
\begin{equation}
  \begin{split}\notag
    \text{(P2)}: \max_{\bm{w},\bm R_0}& \  \min_{l \in \mathcal L} \ \bm{a}^{\mathrm {H}}(\theta_l)\bm{\Phi} \bm{G} (\bm{w} \bm{w}^{\mathrm {H}}\!+ \!\bm R_0 ) \bm{G}^{\mathrm {H}} \bm{\Phi}^{\mathrm {H}}\bm{a}(\theta_l)\\ 
    \text { s.t. }&\quad  \eqref{eq:SNR_cons}\text{,} \ \eqref{eq:energy_cons}\text{,} \ \text{and} \ \eqref{eq:R_0}.
  \end{split}
\end{equation}
Towards this end, we define $\bm{W}=\bm{w}\bm{w}^{\mathrm {H}}$ with $\bm{W} \succeq \bm{0}$ and $\mathrm {rank}(\bm{W}) \le 1$. By letting $\bm h = \bm{G}^{\mathrm {H}} \bm{\Phi}^{\mathrm {H}} \bm{h}_{\text{r}} + \bm{h}_{\text{d}}$ denote the combined channel vector from the BS to the CU and substituting $\bm{W}=\bm{w}\bm{w}^{\mathrm {H}}$, the transmit beamforming optimization in problem $\text{(P2)}$ is reformulated as 
\begin{subequations}
  \begin{align}\notag
    \text{(P2.1)}:\max_{\bm{W},\bm R_0}&\ \ \min_{l \in \mathcal L} \ \bm{a}^{\mathrm {H}}(\theta_l)\bm{\Phi} \bm{G} (\bm{W}+ \bm R_0 ) \bm{G}^{\mathrm {H}} \bm{\Phi}^{\mathrm {H}}\bm{a}(\theta_l)\\  \label{eq:st_1}
    \text { s.t. }& \quad \mathrm {tr}(\bm h \bm h^{\mathrm {H}}\bm{W}) \geq \Gamma \sigma^2,\\ \label{eq:st_2}
    &\quad  \mathrm {tr}(\bm{W}+\bm R_0) \le P_0,\\ \label{eq:st_4}
    &\quad  \bm R_0 \succeq 0, \bm{W} \succeq 0,\\ \label{eq:st_5}
    &\quad \mathrm{rank}(\bm{W}) \le 1.
  \end{align}
\end{subequations}

However, problem $\text{(P2.1)}$ is non-convex due to the rank-one constraint on $\bm{W}$ in \eqref{eq:st_5}. To resolve this issue, we relax the rank-one constraint and accordingly obtain the SDR version of problem $\text{(P2.1)}$ as
\begin{equation}
  \begin{split}\notag
    \text{(SDR2.1)}:\max_{\bm{W},\bm R_0}&\ \ \min_{l \in \mathcal L} \ \bm{a}^{\mathrm {H}}(\theta_l)\bm{\Phi} \bm{G} (\bm{W}+ \bm R_0 ) \bm{G}^{\mathrm {H}} \bm{\Phi}^{\mathrm {H}}\bm{a}(\theta_l)\\  \notag
    \text { s.t. }& \quad \eqref{eq:st_1}\text{,} \ \eqref{eq:st_2}\text{,} \ \text{and} \ \eqref{eq:st_4}.
  \end{split}
\end{equation}
It is observed that problem $\text{(SDR2.1)}$ is a semi-definite program (SDP) that can be solved optimally by convex solvers such as CVX\cite{cvx}. Let $\bm W^*$ and $\bm R_0^*$ denote the obtained optimal solution to problem $\text{(SDR2.1)}$, where $\bm W^*$ is generally of high rank. Based on $\bm W^*$, we can reconstruct an equivalent rank-one solution and accordingly find the optimal solution to problem $\text{(P2.1)}$ (and thus $\text{(P2)}$), as shown in the following proposition.
\begin{proposition} \label{Proposition1}
The optimal solution to problem $\text{(P2.1)}$ is 
\begin{equation}
\hat{\bm W} =\hat{\bm w}\hat{\bm w}^{\mathrm {H}},
\end{equation}
\begin{equation}\label{eq:R_new}
\hat{\bm R_0} = \bm R_0^* +  \bm W^* - \hat{\bm W},
\end{equation}
where $\hat{\bm w} = (\bm h^{\mathrm {H}} \bm W^* \bm h)^{-1/2}\bm W^* \bm h$ denotes the corresponding transmit beamforming vector at the BS. Accordingly, $\hat {\bm w}$ and $\hat {\bm R_0}$ become the optimal solution to problem $\text{(P2)}$.
\end{proposition}
\begin{proof}
 See Appendix A.
\end{proof}

\textit {Remark 2:} 
It is worth noting that for the optimal solution to $\text{(P2.1)}$, we have $\mathrm{rank}(\hat{\bm R}_0) \ge 1$ ($\hat{\bm R}_0 \neq \bm 0$) in general, especially when $\bm G$ is randomly generated (e.g., following Rayleigh fading), as will be shown in numerical results in Section IV. This shows the necessity of using dedicated sensing beams for enhancing the ISAC performance with IRS. This is different from the case without IRS in \cite{hua}, where the dedicated sensing beams may not be needed.

\textit {Remark 3:} 
It is also worth discussing a special case when the BS is deployed with a ULA and the channel matrix from the BS to the IRS is LoS, i.e., 
\begin{equation}
\bm G =\bm a(\theta_\text{IRS}) \bm b^\mathrm{H}(\theta_\text{BS}),
\end{equation}
where $\theta_\text{IRS}$ and $\theta_\text{BS}$ denote the angle of arrive (AoA) and the AoD of the BS-IRS link at the IRS and the BS, respectively, $\bm a(\theta_\text{IRS})$ denotes the steering vector at the IRS in \eqref{equ:steering}, and $\bm b(\theta_\text{BS})$ denotes the steering vector at the BS, which is
\begin{equation} 
  \bm b(\theta_\text{BS}) = [1,e^{j  \frac{2\pi d_{\text{BS}}}{\lambda} \sin \theta_\text{BS}},...,e^{j  \frac{2\pi (M-1) d_{\text{BS}}}{\lambda}  \sin\theta_\text{BS}}]^{\mathrm {T}},
\end{equation} 
with $d_{\text{BS}}$ denotes the spacing between consecutive antennas at the BS.
In this case, the beampattern gain from the IRS towards angle $\theta$ is rewritten as 
\begin{equation}
\mathcal{P}(\theta)=|\bm{a}^{\mathrm {H}}(\theta_\text{IRS}) \bm{\Phi}^{\mathrm {H}} \bm{a}(\theta)|^2\bm b^{\mathrm {H}}(\theta_\text{BS})(\bm{W} + \bm R_0 )\bm b(\theta_\text{BS}).
\end{equation}
In this case, it can be shown similarly as in \cite{hua} that $\bm R_0^* = \bm 0$ is general optimal for problem $\text{(P2.1)}$, which means that the dedicated sensing beams are not necessary in this special case. 

\subsection{Reflective Beamforming Optimization at IRS}
Next, we optimize the reflective beamformer $\bm \Phi$ in problem $\text{(P1)}$ under any given active transmit beamformers $\bm w$ and $\bm R_0$. 
Let $\bm{v} = [e^{j\phi_1},...,e^{j\phi_N}]^{\mathrm {H}}$ denote the reflective phase shift vector at the IRS. Then the sensing beampattern gain from the IRS towards angle $\theta$ is rewritten as
\begin{equation}\label{eq:R_1} 
\mathcal{P}(\theta)=\bm{v}^{\mathrm {H}}\bm{R}_{1}(\theta)\bm{v},
\end{equation}
 where $\bm{R}_{1}(\theta)=\mathrm {diag}(\bm{a}^{\mathrm {H}}(\theta))\bm{G}(\bm R_0 + \bm{W})\bm{G}^{\mathrm {H}} \mathrm {diag}(\bm{a}(\theta))$. Furthermore, define 
\begin{equation} \label{eq:R_2}     
\bm{R}_{2}(\theta)=
\left[                 
  \begin{array}{cc}   
    \bm{R}_{1}(\theta)& \bm{0}_{N\times 1}\\ 
    \bm{0}_{1\times N}& 0\\ 
  \end{array}
\right], 
{\bar{\bm{v}}}=
\left[                 
  \begin{array}{c}   
    \bm{v}\\ 
    1\\ 
  \end{array}
\right].
\end{equation}
By substituting \eqref{eq:R_2} into \eqref{eq:R_1}, we have $\mathcal{P}(\theta)={\bar{\bm{v}}}^{\mathrm {H}}\bm{R}_{2}(\theta)\bar{\bm{v}}$. Furthermore, by letting $\bm{H}=\mathrm {diag}(\bm{h}_{\text{r}}^{\mathrm {H}})\bm G\in \mathbb{C}^{N\times M}$, the received signal power in \eqref{eq:SNR} is rewritten as $|(\bm{h}_{\text{r}}^{\mathrm {H}} \bm{\Phi} \bm{G} + \bm{h}_{\text{d}}^{\mathrm {H}})\bm{w}|^2=|(\bm{v}^{\mathrm {H}}\bm{H}+\bm{h}_{\text{d}}^{\mathrm {H}})\bm{w}|^2$. The SNR constraint in \eqref{eq:SNR_cons} is formulated as
\begin{equation}
    (\bm{v}^{\mathrm {H}}\bm{H}+\bm{h}_{\text{d}}^{\mathrm {H}})\bm{W}(\bm{H}^{\mathrm {H}}\bm{v}+\bm{h}_{\text{d}}) \geq \Gamma \sigma^2,
\end{equation}
which is equivalent to 
\begin{equation}
    {\bar{\bm{v}}}^{\mathrm {H}}\bm{R}_{3}{\bar{\bm{v}}} +\bm{h}_{\text{d}}^{\mathrm {H}}\bm{W}\bm{h}_{\text{d}}\geq \Gamma \sigma^2,
\end{equation}
with \begin{equation}   
\bm{R}_{3}=
\left[                 
  \begin{array}{cc}   
    \bm{H}\bm{W}\bm{H}^{\mathrm {H}} & \bm{H}\bm{W}\bm{h}_{\text{d}}\\ 
    \bm{h}_{\text{d}}^{\mathrm {H}}\bm{W}\bm{H}^{\mathrm {H}} & 0\\ 
  \end{array}
\right].
\end{equation}

As a result, the optimization of $\bm \Phi$ in problem $\text{(P1)}$ becomes  the optimization of $\bar{\bm v}$ in the following problem:
\begin{subequations}
  \begin{align}\notag
    \text{(P3)}:   \max_{\bar{\bm v}}&\ \ \min_{l \in \mathcal L} \ {\bar{\bm{v}}}^{\mathrm {H}}\bm{R}_{2}(\theta_l)\bar{\bm{v}}\\
    \text { s.t. }&\quad {\bar{\bm{v}}}^{\mathrm {H}}\bm{R}_{3}{\bar{\bm{v}}} +\bm{h}_{\text{d}}^{\mathrm {H}}\bm{W}\bm{h}_{\text{d}}\geq \Gamma \sigma^2, \quad\\
    & \quad |\bar{\bm v}_n|=1, \forall n\in \{1,...,N+1\}.
  \end{align}
\end{subequations}
Next, we define $\bar {\bm V}=\bar{\bm v}{\bar{\bm v}}^{\mathrm {H}}$ with $\bar{\bm V} \succeq \bm{0}$ and $\mathrm {rank}(\bar{\bm V})=1$. Note that $\bar{\bm v}^{\mathrm {H}} \bm{R}_{2}(\theta) \bar{\bm v}=\mathrm {tr}( \bm{R}_{2}(\theta) \bar{\bm V})$ and $\bar{\bm v}^{\mathrm {H}} \bm{R}_{3} \bar{\bm v}=\mathrm {tr}( \bm{R}_{3} \bar{\bm V})$.
By substituting $\bar{\bm V}=\bar{\bm v}{\bar{\bm v}}^{\mathrm {H}}$, the reflective beamforming optimization in problem $\text{(P3)}$ is reformulated as
\begin{subequations}
  \begin{align}\notag
    \text{(P3.1)}:   \max_{\bar{\bm V}}&\ \ \min_{l \in \mathcal L} \ \mathrm {tr}(\bm{R}_{2}(\theta_l){\bar {\bm V}})\\ \label{eq:st3_1}
    \text { s.t. }&\quad \mathrm {tr}(\bm{R}_{3}{\bar{\bm V}})+\bm{h}_{\text{d}}^{\mathrm {H}}\bm{W}\bm{h}_{\text{d}} \geq \Gamma \sigma^2, \quad\\ \label{eq:st3_2}
    &\quad  {\bar{\bm V}_{n,n}}=1, \forall n\in \{1,...,N+1\},\\ \label{eq:st3_3}
    & \quad \bar{\bm V} \succeq \bm{0},\\\label{eq:st3_4}
    & \quad \mathrm {rank}(\bar{\bm V})=1.
  \end{align}
\end{subequations}

However, problem $\text{(P3.1)}$ is non-convex due to the rank-one constraint on $\bar{\bm V}$ in \eqref{eq:st3_4}. To resolve this issue, we relax the rank-one constraint and accordingly obtain the SDR version of problem $\text{(P3.1)}$ as 
\begin{equation}
  \begin{split}\notag
    \text{(SDR3.1)}:   \max_{\bar{\bm V}}&\ \ \min_{l \in \mathcal L} \ \mathrm {tr}(\bm{R}_{2}(\theta_l){\bar {\bm V}})\\
    \text { s.t. }&\quad \eqref{eq:st3_1}\text{,} \ \eqref{eq:st3_2}\text{,} \ \text{and} \ \eqref{eq:st3_3}.
  \end{split}
\end{equation}
Problem $\text{(SDR3.1)}$ is an SDP that can be solved optimally by convex solvers such as CVX\cite{cvx}. Let $\bar{\bm V}^*$ denote the obtained optimal solution to problem $\text{(SDR3.1)}$, which is generally of high rank. Then Gaussian randomization is used to construct an approximate rank-one solution. Specifically, we first generate a number of randomizations $\bm r \sim \mathcal{CN}(\bm{0},\bar{\bm V}^*)$, and accordingly construct the candidate feasible solution to problem $\text{(P3)}$ as $\bm{v}=e^{j\mathrm {arg}([\frac{\bm r}{\bm r_{N+1}}]_{(1:N)})}$. 
By independently generating Gaussian random vector $\bm{r}$ multiple times, the objective value is approximated as the maximum one among all these random realizations. Notice that the Gaussian randomization should be implemented many times to ensure that the objective value increase at each iteration. 

By combining the obtained solutions to problems $\text{(P2.1)}$ and $\text{(P3.1)}$, the alternating optimization based algorithm for solving problem $\text{(P1)}$ is complete. Notice that in each iteration of alternating optimization, problem $\text{(P2.1)}$ is optimally solved, and thus it will result in a non-decreasing objective value. Also notice that with sufficient number of randomizations, the objective value after solving problem $\text{(P3.1)}$ will be monotonically non-decreasing as well. As a result, the convergence of the proposed alternating optimization based algorithm for solving problem $\text{(P1)}$ is ensured. 

\section{Numerical Results}
This section provides numerical results to evaluate the performance of our proposed IRS-assisted ISAC design. We consider Rician fading for the BS-IRS and IRS-CU links with the Rician factor being $0.5$, and Rayleigh fading for the BS-CU link. The distance-dependent path loss is modeled as $K_0(\frac{d}{d_0})^{-\alpha}$, where $K_0=-30~ \text{dB}$ is the path loss at the reference distance $d_0=1~ \text{m}$ and the path-loss exponent $\alpha$ is set as $2.5$, $2.5$, and $3.5$ for the BS-IRS, IRS-CU, and BS-CU links, respectively. Due to the potential obstacles, we consider additional shadow fading for the BS-CU link, with a standard deviation of $10~ \text{dB}$. The BS, CU, and IRS are located at coordinate $(0,0)$, $(50~\text{m},0)$, and $(18~\text{m},2~\text{m})$, respectively. 
The desired sensing angles are sampled over $[-61^{\circ}, -59^{\circ}]$, $[-31^{\circ},-29^{\circ}]$, $[-1^{\circ},1^{\circ}]$, $[29^{\circ},31^{\circ}]$, and $[59^{\circ},61^{\circ}]$, where the sampling interval is $0.25^{\circ}$. We  set the number of antennas at the BS and the number of reflecting elements at the IRS as $M = 8$ and $N = 64$, respectively. We also set  $P_0 = 20~ \text{dBm}$ and $\sigma^2 = -80~ \text{dBm}$.

Fig.~\ref{convergence} shows the convergence behavior of our proposed alternating optimization based algorithm for solving problem $\text{(P1)}$, where $\Gamma=10~ \text{dB}$. It is shown that the proposed alternating optimization based
algorithm converges within around $10$ iterations, thus validating its effectiveness.

Next, we compare the performance of our proposed IRS-assisted ISAC design, versus the following benchmark schemes.

\subsubsection{Information beamforming only} The BS uses information beam $\bm w$ for both communication and sensing. The joint beamforming design in this case corresponds to problem $\text{(P1)}$ with $\bm R_0= \bm 0$. In this scheme, the obtained solution of problem $\text{(SDR2.1)}$ is usually with high rank and Gaussian randomization is used to reconstruct an approximate rank-one solution.

\subsubsection{Separate beamforming design}
This scheme optimizes the transmit beamforming at the BS and reflection beamforming at the IRS, respectively. First, the reflection beamformer $\bm \Phi$ is optimized to maximize the IRS’s minimum channel norms towards the desired sensing angles, i.e., 
\begin{equation}\label{eq:P4}
\begin{split}\notag
  \text{(P4)}: &\max_{\bm{\Phi}}\ \ \min_{l \in \mathcal L} \ \mathbb{E}(\|\bm{a}^{\mathrm {H}}(\theta_l)\bm{\Phi} \bm{G} \|^2)\\
 &\text { s.t. } \quad  \bm{\Phi} = \mathrm {diag}(e^{j\phi_1},...,e^{j\phi_N}).
\end{split}
\end{equation}
Then, with  the reflective beamformer $\bm \Phi$ obtained from problem $\text{(P4)}$, the joint transmit information and sensing  beamforming is designed by solving problem $\text{(P2)}$.

\begin{figure}[t]
    \centering
    \includegraphics[width=0.35\textwidth]{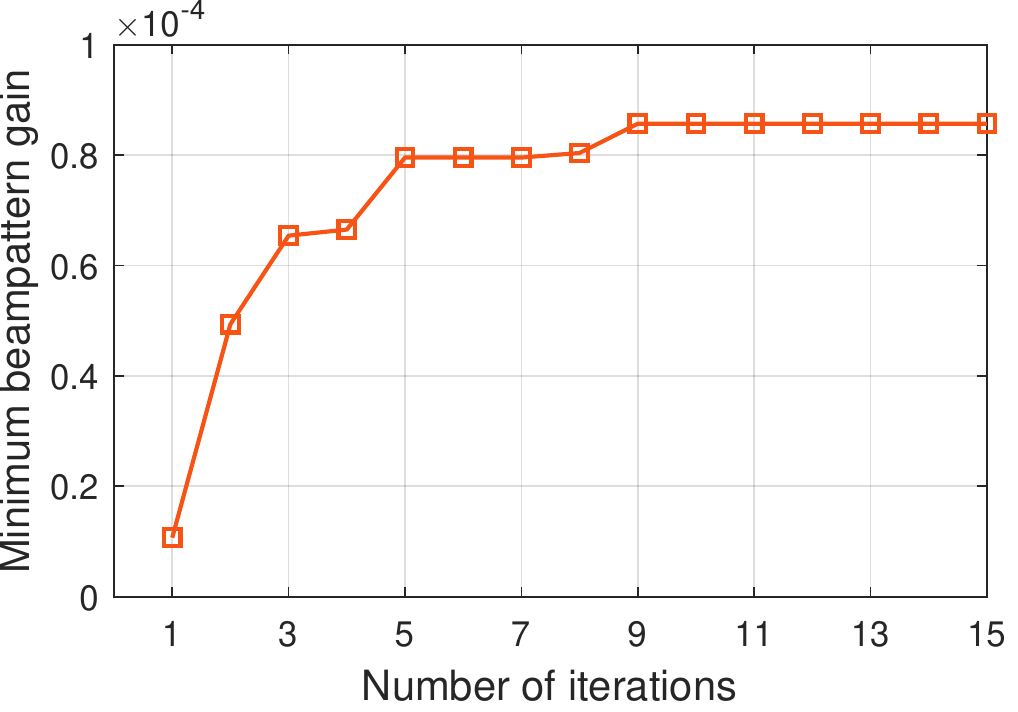}
    \caption{Convergence behavior of the proposed alternating optimization based algorithm for solving problem $\text{(P1)}$, where $\Gamma=10$ dB.}
    \label{convergence}
\end{figure}
\begin{figure}[t]
    \centering
    \includegraphics[width=0.42\textwidth]{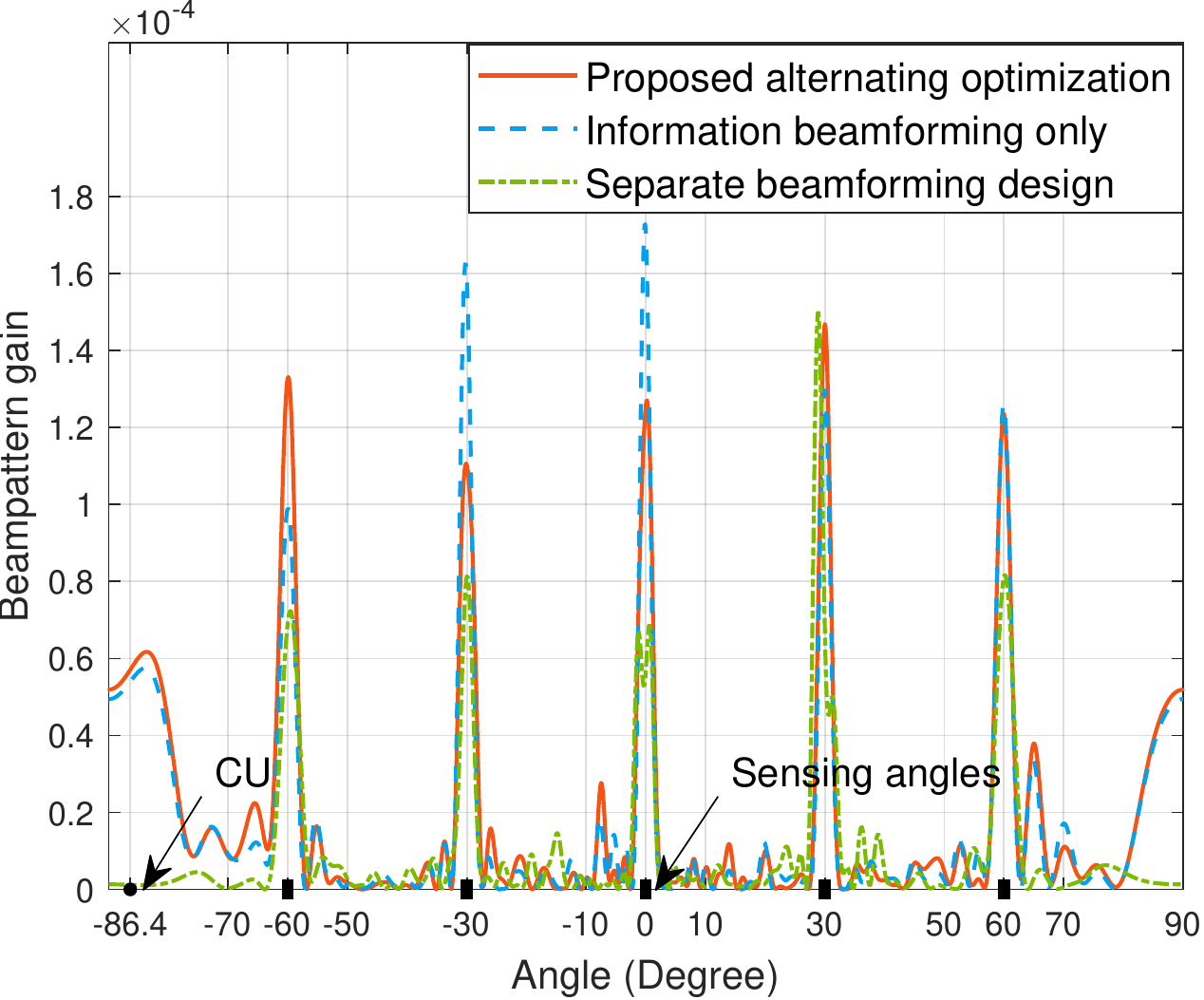}
    \caption{The beampattern gains achieved by different schemes, where $\Gamma=10$ dB.}
    \label{beampattern}
\end{figure}
\begin{figure}[t]
    \centering
    \includegraphics[width=0.35\textwidth]{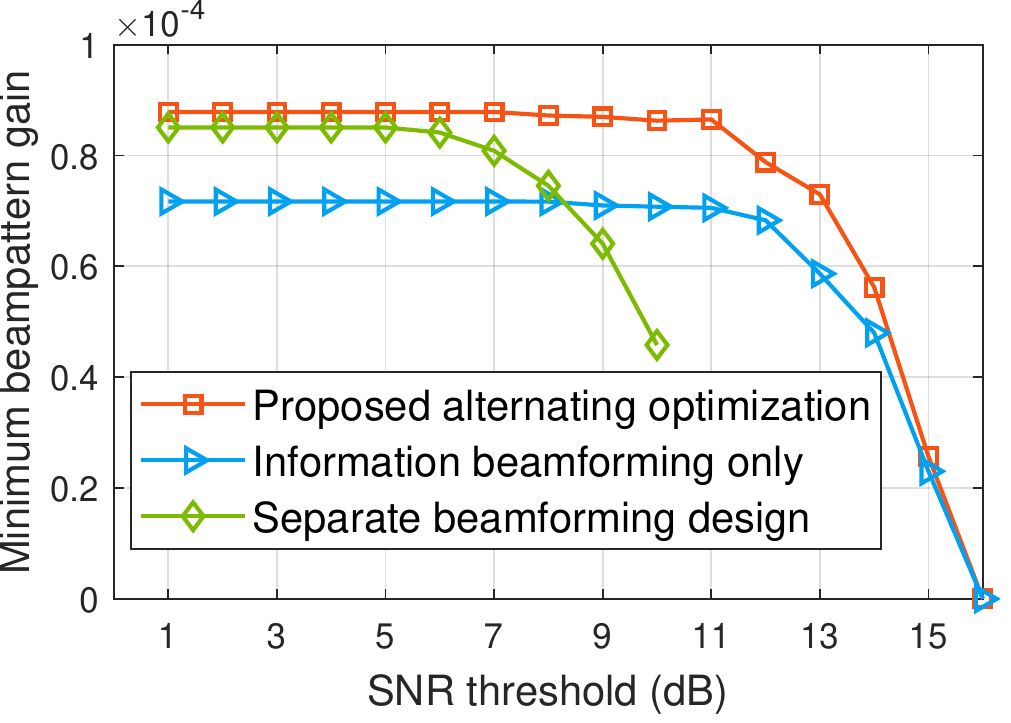}
    \caption{The minimum beampattern gains at the desired sensing angles versus the SNR threshold $\Gamma$ at the CU.}
    \label{gain_SNR_dB}
\end{figure}

Fig.~\ref{beampattern} shows the achievable beampattern gains obtained by different schemes, where $\Gamma=10~ \text{dB}$. 
It is observed that the proposed designs generate multiple signal beams from the IRS towards the desired sensing angles. The proposed alternating optimization based design is observed to achieve higher beampattern gains at the worst angles, than the other two benchmark schemes.

Fig.~\ref{gain_SNR_dB} shows the minimum beampattern gains at the desired sensing angles versus the received SNR threshold $\Gamma$ at the CU. It is observed that the proposed alternating optimization based algorithm achieves significant performance gains over the two benchmark schemes. When $\Gamma$ is small (e.g., $\Gamma < 6 ~ \text{dB}$), the separate beamforming design is observed to perform close to the proposed alternating optimization based design and outperform the information beamforming only scheme. This is due to the fact that in this case the SNR at the CU can be easily satisfied, and thus the sensing oriented reflective beamforming (i.e, problem $\text{(P4)}$) in the separate beamforming design becomes desirable. Furthermore, it is observed that the employment of sensing signal leads to significant sensing performance enhancement, as compared to the counterpart with information beamforming only. This is consistent with \textit{Remark 2} due to the consideration of random channel $\bm G$. By contrast, when $\Gamma$ is high (e.g., $\Gamma >12~ \text{dB}$), it is observed that the information only design performs close to the alternating optimization based design, and the separate beamforming design becomes  infeasible.  This is because that in this case, most energy should be allocated for information signals to meet the SNR requirement at the CU, thus making the information beamforming only design favorable.

\section{Conclusion}
This paper studied the joint transmit and reflective beamforming design in an IRS-assisted ISAC system with a
CU and multiple potential sensing targets at the NLoS areas of the BS. The IRS was deployed to not only assist the communication, but also create virtual LoS links for sensing targets in those conventionally NLoS covered areas. Our objective was to maximize the IRS’s minimum beampattern gain towards the desired sensing angles, while ensuring the SNR requirement at the CU, by jointly optimizing the transmit information and sensing beamforming at the BS and the reflective beamforming at the IRS. To solve this non-convex problem, we proposed an efficient algorithm based on the alternating optimization and SDR. Numerical results showed that the proposed algorithm achieves improved beampattern gains towards desired angles to enable IRS-assisted sensing, while ensuring the communication requirement. It is also shown that our proposed design with dedicated sensing beams significantly outperforms the benchmark scheme without using sensing beams (by only reusing information beams for sensing), which reveals the importance of dedicated sensing beams in IRS-assisted ISAC systems. 

\appendices
\section{Proof of Proposition 1}
It follows from \eqref{eq:R_new} that $ \hat{\bm W} +\hat{\bm R_0} =  \bm W^*+\bm R_0^*$, and as a result, $\hat{\bm W}, \hat{\bm R_0}$ and $\bm W^*, \bm R_0^*$ achieve the same objective values and both satisfy the constraint in \eqref{eq:st_2}. Next, it can be verified that
\begin{equation}
\mathrm {tr}(\bm h \bm h^{\mathrm {H}} \hat{\bm W})
=\mathrm {tr}(\bm h \bm h^{\mathrm {H}} \bm W^*).
\end{equation}
Therefore, $\hat{\bm W}$ also satisfies the SNR constraint in \eqref{eq:st_1}.
Furthermore, for any $\bm y \in \mathbb{C}^{M \times 1}$, it holds that
\begin{equation}
\bm y^{\mathrm {H}} (\bm W^* - \hat{\bm W})\bm y = \bm y^{\mathrm {H}} \bm W^* \bm y - |\bm y^{\mathrm {H}} \bm W^* \bm h|^2(\bm h^{\mathrm {H}} \bm W^* \bm h)^{-1}.
\end{equation}
According to the Cauchy-Schwarz inequality, we have
\begin{equation}
(\bm y^{\mathrm {H}} \bm W^* \bm y) (\bm h^{\mathrm {H}} \bm W^* \bm h)   \geq |\bm y^{\mathrm {H}} \bm W^* \bm h|^2,
\end{equation}
and it follows that $\bm y^{\mathrm {H}} (\bm W^* - \hat{\bm W})\bm y\geq 0$. Accordingly, we have  $\bm W^* - \hat{\bm W} \succeq 0$. In addition, as  the summation of a set of positive semidefinite matrices is also
positive semidefinite, it follows that $\hat{\bm R_0} \succeq 0$. Hence, $\hat{\bm W}, \hat{\bm R_0}$ also satisfy the constraint in \eqref{eq:st_4}. Notice that $\mathrm {rank} (\hat{\bm W}) \le 1$ with $\hat{\bm W} =\hat{\bm w}\hat{\bm w}^{\mathrm {H}}$. Therefore, $\hat{\bm W}$ and $\hat{\bm R_0}$ are optimal for problem $\text{(P2.1)}$. Proposition $1$ is finally proved.

\ifCLASSOPTIONcaptionsoff
  \newpage
\fi

\bibliographystyle{IEEETran}
\bibliography{IEEEabrv,myref}

\end{document}